\documentclass[letterpaper, 9 pt, conference]{ieeeconf}      % Use this line for a4
\IEEEoverridecommandlockouts                              % This command is only
\usepackage{graphics} % for pdf, bitmapped graphics files
\usepackage{epsfig} % for postscript graphics files
\usepackage{times} % assumes new font selection scheme installed
\usepackage{amsmath} % assumes amsmath package installed
\usepackage{amssymb}  % assumes amsmath package installed
\usepackage{subfigure}
\usepackage{url}
\usepackage{dsfont}
\usepackage{mathrsfs}
\usepackage{bbm}
\newtheorem{theorem}{Theorem}

\usepackage{cite}
\newtheorem{assumption}[theorem]{Assumption}
\usepackage{color}
\usepackage{algorithm}

\newtheorem{corollary}[theorem]{Corollary}

\newtheorem{definition}[theorem]{Definition}

\newtheorem{lemma}[theorem]{Lemma}

{ % these settings define remarks and example as parts
\newtheorem{remark}[theorem]{Remark}

}

\newcommand{\bi}{\begin{itemize}}
\newcommand{\ei}{\end{itemize}}
\newcommand{\bd}{\begin{displaymath}}
\newcommand{\ed}{\end{displaymath}}
\newcommand{\be}{\begin{eqnarray*}}
\newcommand{\ee}{\end{eqnarray*}}

\usepackage[normalem]{ulem}

\title{\LARGE \bf
Data-driven Stabilization of Discrete-time Control-affine Nonlinear Systems: A Koopman Operator Approach}
\author{Subhrajit Sinha, Sai Pushpak Nandanoori, Jan Drgona,  Draguna Vrabie\\
\thanks{This research was supported by the Assistant Secretary for Energy Efficiency and Renewable Energy, Office of Building Technologies of the U.S. Department of Energy, under Contracts No. DE-AC05-76RL01830.
\newline
All the authors are with Pacific Northwest National Laboratory, Richland, WA, USA - 99354.
Emails: \{subhrajit.sinha, saipushpak.n, jan.drgona, draguna.vrabie\}@pnnl.gov 
}
}

\begin{document}
\maketitle

\begin{abstract}
In recent years data-driven analysis of dynamical systems has attracted a lot of attention and transfer operator techniques, namely, Perron-Frobenius and Koopman operators are being used almost ubiquitously. Since data is always obtained in discrete-time, in this paper, we propose a purely data-driven approach for the design of a stabilizing feedback control law for a general class of discrete-time control-affine non-linear systems. In particular, we use the Koopman operator to lift a control-affine system to a higher-dimensional space, where the control system's evolution is bilinear. We analyze the controllability of the lifted bilinear system and relate it to the controllability of the underlying non-linear system. We then leverage the concept of Control Lyapunov Function (CLF) to design a state feedback law that stabilizes the origin. Furthermore, we demonstrate the efficacy of the proposed method to stabilize the origin of the Van der Pol oscillator and the chaotic Henon map from the time-series data.
\end{abstract}

\section{Introduction}
The design of stabilizing control law for general nonlinear systems is a central problem in control theory with applications in many different branches of engineering like power systems, biological networks, building systems, etc. To this end, the Sum of Squares (SOS) approach \cite{SOS, parrilo_SOS} and differential geometry-based feedback linearization \cite{sontag_feedback} techniques provide systematic approaches for the design of stabilizing feedback control laws for general nonlinear systems. Another approach for the design of stabilizing control laws stems from the ergodic theory of dynamical systems \cite{Lasota}. One of the main advantages of using ergodic and operator theoretic ideas is that these expositions generate \emph{exact linear representations} of the underlying nonlinear systems. Thus, one can leverage concepts from linear control theory for the analysis and control of nonlinear systems. Furthermore, they provide an opportunity for data-driven analysis and control of dynamical systems.

Motivated by these advantages there has been growing interest in transfer operator theoretic techniques, namely Perron-Frobenius and Koopman operator techniques, for analysis of dynamical systems \cite{Mezic2000,Vaidya_TAC, EDMD_williams,sinha_sparse_koopman_acc,mezic_koopmanism,sinha_equivariant_IFAC,sai_global_koopman_acc,nandanoori2021data,sinha_few_shot_arxiv,sinha_IT_CDC_2016,sinha_IT_ICC,sinha_IT_power_CDC,sinha_IT_power_journal}. Building on initial works of data-driven learning of dynamical systems using Koopman operators, in 
\cite{robust_DMD_journal,subspace_koompan} the authors provided algorithms for computation of Robust Koopman operator from noisy data. Moreover, \cite{lusch2018deep, yeung2019learning, nandanoori2022graph} used deep learning techniques for learning both the observable functions and the Koopman operator from time-series data.

In the application front, \cite{surana_cdc_2016} used Koopman operators to design observers for general nonlinear systems, for control of non-equilibrium dynamics \cite{optimal_placement_JMAA,sinha_IT_optimal_placement_ICC}, causal analysis and topology identification in dynamical networks \cite{sinha_IT_data_acc,sinha_IT_data_journal}, analysis of power networks \cite{susuki2013nonlinear, sinha_computationally_efficient,sinha_online_PES}, for attack detection in power networks \cite{nandanoori2020model}, learning and analysis of biological systems \cite{sinha_genetic} etc. 

For dynamical systems with control, in \cite{kaiser_control,proctor_control}, the authors proposed a method for computation of Koopman operators in the presence of control inputs. \cite{korda_mezic_predictor} discussed how Koopman operators could be used for model predictive control. In \cite{bowen_feedback, bowen_optimal, bowen_convex} the Koopman operator framework was used for designing stabilizing and optimal control laws for a class of continuous-time dynamical systems. However, in all practical applications, data is always obtained as a discrete set, and thus it is natural to design stabilizing feedback control law from a discrete-time systems point of view. 

In this paper, we address the problem of designing stabilizing state feedback control law for discrete-time control-affine systems. In particular, we use the Koopman operator framework to lift the control-affine system to a higher dimensional space where the Koopman representation of the underlying control-affine system is bilinear. We first analyze the global controllability of the lifted bilinear system and relate the controllability of the control-affine system with those of the lifted bilinear system. Moreover, we design a stabilizing state feedback control law for this lifted bilinear control system and show that the control law quadratically stabilizes the origin of the lifted bilinear system. Subsequently, the closed-loop trajectories in the lifted space are then projected back to the state space so that the original system's origin is stabilized.

The rest of the paper is organized as follows. In section \ref{section_transfer_operators} we discuss the basics of transfer operator theory. The paper's main results are presented in section \ref{section_nonlinear}, where we analyze the controllability of the nonlinear system and relate it to the controllability of the Koopman lifted system and propose an optimization problem for the design of the stabilizing feedback control law. Next in section \ref{section_DMD} we briefly discuss the data-driven computation of the Koopman lifted system followed by simulation results in section \ref{section_simulation}, where we demonstrate the efficacy of the proposed approach on two typical nonlinear systems. Finally we conclude the paper in section \ref{section_conclusion}.

\section{Preliminaries}\label{section_transfer_operators}
In this section, we present some preliminaries on transfer operators for a discrete-time dynamical system. 
Consider a discrete-time dynamical system
\begin{eqnarray}\label{system}
z_{t+1} = T(z_t)
\end{eqnarray}
where $T:Z\subset \mathbb{R}^N \to Z$ is assumed to be of class at least ${\cal C}^2$. The dynamical system (\ref{system}) can also be formally defined as a tuple $(Z,{\cal B}(Z), \mu , T)$, where ${\cal B}(Z)$ is the Borel $\sigma$-algebra on $Z$, $\mu$ is a finite measure on the $\sigma$-algebra and $T:Z \rightarrow Z$ is a $\cal B$-measurable map which governs the evolution of the states $z\in Z$. With this, associated with the dynamical system $(Z,{\cal B}(Z), \mu , T)$, one can define two operators, namely, the Perron-Frobenius (P-F) operator and the Koopman operator which governs the evolution of measures \footnote{With a slight abuse of notation we consider the P-F operator to propagate measures. See \cite{Vaidya_TAC}} and functions, under the map $T$, respectively \cite{Lasota}.

\begin{definition}[Perron-Frobenius Operator \cite{Lasota}]
Let $(Z,{\cal B}(Z), \mu , T)$ be a discrete-time dynamical system and let ${\cal M}(Z)$ be the vector-space of signed measures on $Z$. Then the Perron-Frobenius operator $\mathbb{P}:{\cal M}(Z)\to {\cal M}(Z)$ is given by
\begin{eqnarray*}
[\mathbb{P}\mu](A)=\int_{{\cal Z} }\delta_{T(z)}(A)d\mu(z)=\mu(T^{-1}(A)),
\end{eqnarray*}
where $\delta_{T(z)}(A)$ is stochastic transition function which measure the probability that a point $z$ will reach the set $A$ in one time step under the system mapping $T$. 
\end{definition}

As mentioned earlier, there is another operator, namely, the Koopman operator, which governs the evolution of functions under the map $T$ and it is defined as follows:
\begin{definition} [Koopman Operator \cite{Lasota}] 
Given any $h\in\cal{F}$, the Koopman operator $\mathbb{U}:{\cal F}\to {\cal F}$ is defined as
\[[\mathbb{U} h](z)=h(T(z)),\]
where $\cal F$ is the space of functions (observables) invariant under the action of the Koopman operator $\mathbb{U}$.
\end{definition}

\begin{figure}
\centering
\includegraphics[scale=.275]{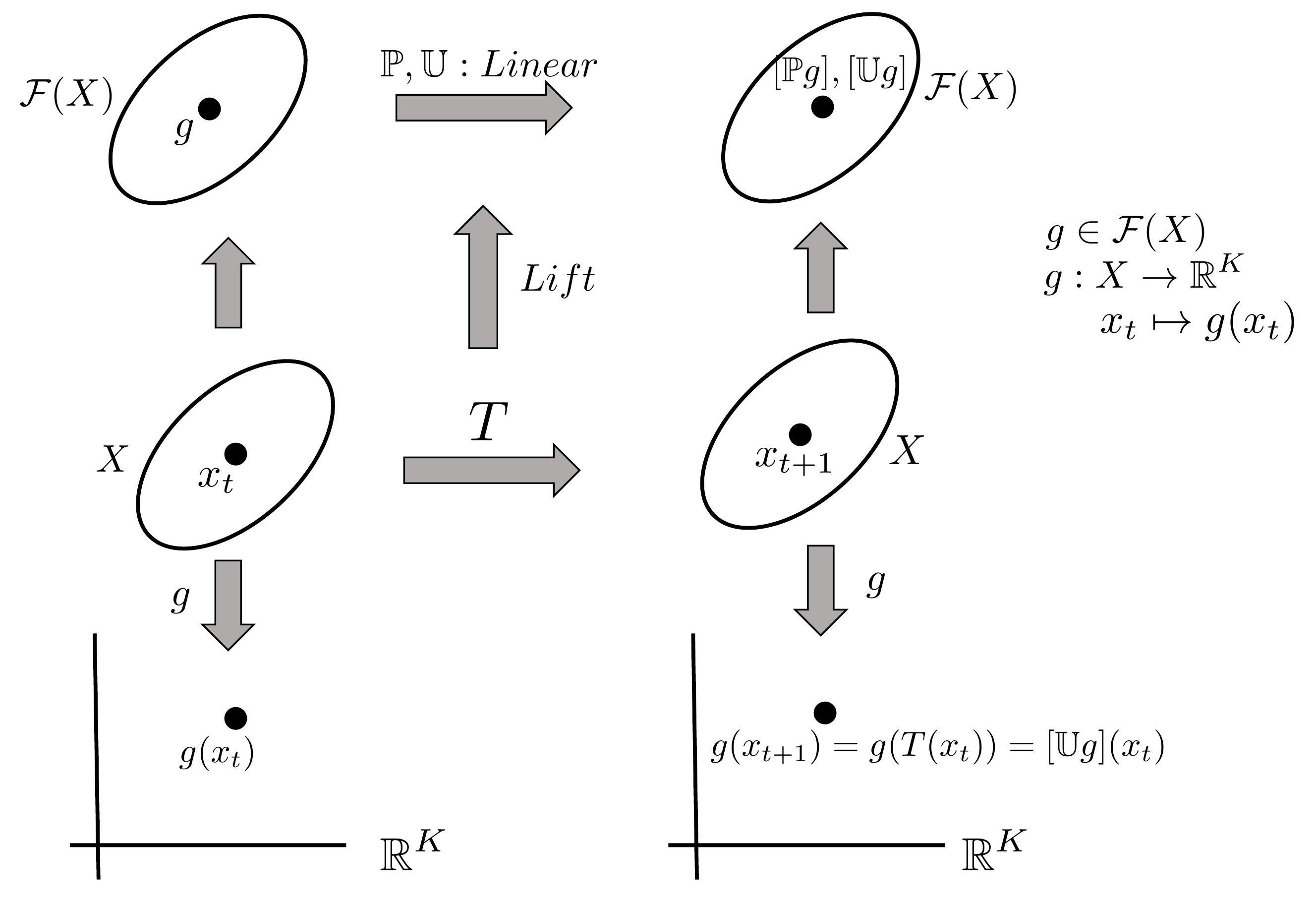}
\caption{Koopman and Perron-Frobenius operators.}\label{koopman_diagram}
\end{figure}

It can be shown that the P-F and the Koopman operators are dual to each other \cite{Lasota} in the sense that 
\[\langle \mathbb{U}f,g\rangle = \langle f,\mathbb{P}g\rangle,\]
for $f\in L_{\infty}(Z)$ and $g\in L_1 (Z)$.

The advantage of using the above operator theoretic approach is the fact that both the P-F operator and the Koopman operator are linear operators, even if the underlying system is non-linear. However, though the operators are linear, the trade-off is that these operators are typically infinite-dimensional. In particular, the P-F operator and the Koopman operator often lift a nonlinear dynamical system from a finite-dimensional space to generate an infinite-dimensional linear system in infinite dimensions. 

\begin{definition}[Koopman Eigenfunction (KEF)] An eigenfunction of the Koopman operator $\mathbb{U}$ is a non-zero observable $\phi_{\lambda} \in {\cal F}$ that satisfies
\begin{align*}
{\mathbb{U}} \phi_{\lambda} = \lambda \phi_{\lambda},
\end{align*}
where $\lambda \in \mathbb{C}$ is the Koopman eigenvalue (KE) corresponding to the KEF $\phi_{\lambda}$. 
\end{definition}

% The KEs and KEFs satisfy the following properties: 
% % 
% \begin{enumerate}
%     \item Suppose, $\phi_{\lambda_1}$ and $\phi_{\lambda_2}$ are two KEF associated with the KEs $\lambda_1$ and $\lambda_2$. If $\phi_{\lambda_1}^{k_1} \phi_{\lambda_2}^{k_2} \in {\cal F}$ for $k_1, k_2 \in \mathbb{R}$, then it is an eigenfunction associated with the eigenvalue $k_1 \lambda_1 + k_2 \lambda_2$ \cite{mauroy2016global,surana_cdc_2016,budivsic2012applied}. 
%     \item KEFs are smooth in the vicinity of the attractor \cite{mauroy2016global,surana_cdc_2016}. 
%     \item The KEFs form an Abelian semigroup and when $\phi_{\lambda}$ vanishes nowhere, then $\phi_{\lambda}^{-1} (x) = 1/\phi_{\lambda} (x)$, and the KEFs form an Abelian group \cite{budivsic2012applied}. 
% \end{enumerate}

% \begin{remark}
% In the subsequent discussions, we assume that the autonomous nonlinear system has a unique equilibrium point. If this is not the case, then one needs to partition the state space into the domain of attractions corresponding to each equilibrium point and carry out the analysis on each of the partitions separately.
% \end{remark}

\section{Nonlinear Stabilization using Koopman Operator}\label{section_nonlinear}
In this section, we present the main theoretical results of the paper where we formulate an optimization problem for the design of state feedback stabilizing control law for a discrete-time control-affine system. For simplicity of presentation, we consider the case of a single input control system.

\subsection{Koopman Representation of the Nonlinear Control System}
We consider a discrete-time control-affine system of the form 
\begin{equation}\label{affine_sys}
x_{t+1} = T(x_t) + g(x_t)u_t
\end{equation}
where $x_t\in\mathbb{R}^d$, $u_t\in\mathbb{R}$ is the single input to the system and $T,g:\mathbb{R}^d\to\mathbb{R}^d$ are at least of class ${\cal C}^2$. Let $(\lambda_i,\phi_i)$, $i=1,2,\dots$ be the discrete eigenvalues and the corresponding eigenfunctions of the Koopman operator $\mathbb{U}$ for the unactuated system $x_{t+1} = T(x_t)$.

\begin{assumption}\label{assumption_span}
Let ${\cal F}^n = \textnormal{span}\{\phi_i\}_{i=1}^n$ be the span of a finite subset of the Koopman eigenfunctions such that $x \in {\cal F}^n$. Hence we have
\begin{eqnarray}\label{assumption_1_eq}
x = \sum_{i=1}^n \phi_i(x){\bf v}_i^x,
\end{eqnarray}
where ${\bf v}_i^x\in\mathbb{C}^d$ are the Koopman modes.
\end{assumption}

We follow the notation of \cite{surana_cdc_2016} and order the Koopman eigenfunctions $\{\phi_1, \cdots, \phi_n\}$ and the corresponding Koopman eigenvalues and Koopman modes such that the complex conjugate pairs are placed adjacent to each other. Let $\Phi(x) = \{\hat{\phi}_1(x),\cdots , \hat{\phi}_n(x)\}$, such that
\begin{enumerate}
\item{$\hat{\phi}(x) = \phi_i(x)$ if the $i^{th}$ Koopman eigenfunction is real, and}
\item{$\hat{\phi}_i(x) = 2\textnormal{Re}(\phi_i)$ and $\hat{\phi}_{i+1}(x) = -2\textnormal{Im}(\phi_i)$, if $i^{th}$ and $(i+1)^{th}$ Koopman eigenfunctions are complex conjugate pairs.}
\end{enumerate}
where $\textnormal{Re}(\cdot)$, $\textnormal{Im}(\cdot)$ denotes the real and imaginary part respectively. The Koopman mode decomposition (\ref{assumption_1_eq}) can now be written as
\begin{align}
x = C^x\Phi(x),
\label{proj_maps}
\end{align}
where $C^x\in\mathbb{C}^{d\times n}$. Note that, since $x = C^x\Phi(x)$, the mapping $\Phi$ is injective onto its range and $\Phi(0)=0$. Moreover, the matrix $C^x$, which can be obtained from data as a solution to a least squares problem, is used to project the lifted trajectories $(z_t\in\mathbb{R}^n)$ back to the state space $(x_t\in\mathbb{R}^d)$.

We further assume the following.
\begin{assumption}\label{assumption_1}
We assume that $\frac{\partial {\bf \Psi}}{\partial x}g$ lies in ${\cal F}^n$, so that there exists a constant matrix $B\in\mathbb{R}^{n\times n}$ such that $\frac{\partial {\bf \Psi}}{\partial x}g=B{\bf \Psi}$.
\end{assumption}

\begin{remark}
On the face of it, assumption \ref{assumption_1} seems a bit restrictive and depends on the the functions $\hat{\phi}_i$ and $g$. In the cases where assumption \ref{assumption_1} fails to hold, one can obtain the matrix $B$ as a solution of a least-squares problem where we have $\frac{\partial {\bf \Psi}}{\partial x}g = B\Phi + \varepsilon$, for $\varepsilon \in\mathbb{R}^n$, so that we minimize the norm of the error vector $\varepsilon$ to find the optimal $B$.
\end{remark}

%For further justification of the above assumption \ref{assumption_1}, see \cite{bowen_feedback}.

\begin{lemma}\label{lemma_koopman_representation}
Under assumption \ref{assumption_1} and the set of observables (lifting functions) defined as $\Phi(x_t)=z_t$, the Koopman representation of the control-affine system (\ref{affine_sys}) is given by
\begin{eqnarray}\label{Koopman_representation}
z_{t+1} = \mathbb{U}z_t + u_t B z_t.
\end{eqnarray}
\end{lemma}
\begin{proof}
Let $z(t)=\Phi(x(t))$. Then we have,
\begin{eqnarray}\label{Koopman_system}\nonumber
z_{t+1} &=& \Phi(x_{t+1}) = \Phi\Big( T(x_t) + g(x_t)u_t\Big)\\ \nonumber
&\approx& \Phi\big(T(x_t)\Big) + \frac{\partial \Phi}{\partial x} g(x_t)u_t + h.o.t.\\ \nonumber
&=& \mathbb{U}\Phi(x_t) + \frac{\partial \Phi}{\partial x} g(x_t)u_t\\
&=& \mathbb{U}z_t + \frac{\partial \Phi}{\partial x} g(x_t)u_t
\end{eqnarray}

Now from assumption \ref{assumption_1}, $\frac{\partial \Phi}{\partial x} g(x_t) = B\Phi(x_t)= Bz_t$ and hence the Koopman representation (\ref{Koopman_system}) of the nonlinear system (\ref{affine_sys}) becomes
\begin{eqnarray}\label{Koopman_representation}
z_{t+1} = \mathbb{U}z_t + u_t B z_t.
\end{eqnarray}

Hence, for the control-affine system (\ref{affine_sys}), the equivalent Koopman representation is a bilinear control system of the form (\ref{Koopman_representation}).
\end{proof}

% Note that, under assumptions \ref{assumption_span} and \ref{assumption_1}, the Koopman representation (\ref{Koopman_representation}) is a finite-dimensional exact representation of the control-affine nonlinear system (\ref{affine_sys}).
Note that under assumptions \ref{assumption_span} and \ref{assumption_1}, the Koopman representation (\ref{Koopman_representation}) is a finite-dimensional exact representation of the control-affine system (\ref{affine_sys}). See \cite{surana_cdc_2016} for the continuous-time counterpart of this.

% Given a nonlinear system, the Koopman operator generates an equivalent linear representation of the nonlinear system. But typically, the representation is in infinite dimensions. However, it can be seen that under assumption \ref{assumption_span}, one obtains a finite-dimensional linear representation of the nonlinear system. With this, we define a new set of co-ordinates $z(t) = \Phi(x(t))$, such that we have

\subsection{Controllability of the original and Koopman lifted system}

The concepts of controllability and observability are fundamental notions in systems theory. But for general nonlinear systems, it is not always easy to determine whether it is controllable (observable). Moreover, in the case of nonlinear systems, there are multiple notions of controllability like accessibility, local controllability and global controllability \cite{bullo_book}. For the completeness of the paper, we revisit some of the basic definitions related to controllability of nonlinear systems.

\begin{definition}[Accessibility \cite{bullo_book}]
A nonlinear system is accessible from the state $x_0\in\mathbb{R}^d$ if the attainable set (both in forward and backward time) from $x_0$ has a non-empty interior. 
\end{definition}

Hence if a system is accessible from some $x_0$, it implies that starting from $x_0$ one can drive the system to some open subset of the configuration manifold. However, this does not imply controllability or even local controllability.

\begin{definition}[Local Controllability \cite{bullo_book}]
A nonlinear system is locally controllable from $x_0\in\mathbb{R}^d$ if the reachable set from $x_0$ contains a neighbourhood of $x_0$.
\end{definition}
In other words, local controllability implies that starting from $x_0$ the system can be driven to any point in the state space which is \emph{near} $x_0$. 

\begin{theorem}\label{theorem_local_controllability}
The control-affine system (\ref{affine_sys}) is locally controllable at $x_0$ if the Lie algebra of the vector field at $x_0$ span the tangent space at $x_0$, that is the rank of the Lie algebra generated by the vector field of (\ref{affine_sys}) at $x_0$ is $d$, where $d$ is the dimension of the configuration manifold.
\end{theorem}
\begin{proof}
See \cite{bullo_book,discrete_time_controllability}.
\end{proof}

In the case of linear systems, local controllability implies that the system is globally controllable, that is, if the trajectories from $x_0$ can be steered to all points of the state space that are in the local neighbourhood of $x_0$, it implies that the system can be driven to any point in the state space from any other point. However, for general nonlinear systems, this is not the case and the notion of global controllability for nonlinear systems is defined as follows.

\begin{definition}[Global Controllability \cite{bullo_book}]\label{def_global_controllability}
If a nonlinear system is locally controllable for all $x_0$ of the configuration manifold, then the system is globally controllable.
\end{definition}
Note that the above definition does imply global controllability because a trajectory starting from any $x_0$ can be steered to any point in the state space by \emph{patching} together trajectories obtained from local controllability and this patching can be achieved because the nonlinear system is locally controllable for all points in the state space.

We refer the reader to \cite{bullo_book,discrete_time_controllability} for further details.

\subsubsection{Controllability of general bilinear systems} Consider a general bilinear system of the form 
\begin{eqnarray}\label{bilinear_general}
z(t+1) = Az(t) + \sum_{i = 1}^l u_i(t)B^iz(t) + B^0 u_0(t),
\end{eqnarray}
where $z(t)\in\mathbb{R}^n$ are the states, $u_i(t)$ are scalar inputs with input matrices $B^i$. For the bilinear system (\ref{bilinear_general}), the drift vector field is $f_d = A z$ and define $\theta_0(z)=B^0$ and $\theta_i(z) = B^iz$ for $i = 1,2,\dots , l$. 

\begin{theorem}\label{theorem_bilinear_controllability}
The bilinear system (\ref{bilinear_general}) is globally controllable if the accessibility distribution
\begin{eqnarray*}
&&{\cal Q} = \Big[B^0, B^1, \dots , B^l, -AB^0, \dots , (-1)^{n-1}A^{n-1}B^0,\\
&& B^1B^0,\dots , B^lB^0, [B^1,A]z,[B^2,A]z,\dots , \big[[B^l,A]A^{n-1}\big]z \Big]
\end{eqnarray*}
has full rank, that is rank$({\cal Q})=n$ for all $z\in\mathbb{R}^n$.
\end{theorem}
\begin{proof}
Local controllability of a general nonlinear system is characterized by the accessibility distribution ${\cal Q} = \textnormal{Lie}(\{X_1,X_2,\dots , X_M\})$, where $X_i$ are the vector fields for the nonlinear system. Hence, for analyzing the local controllability of the bilinear system (\ref{bilinear_general}), we need to compute the Lie brackets among the different vector fields. We have $f_d = Az$, $\theta_0(z)=B^0$ and $\theta_i(z) = B^iz$ for $i = 1,2,\dots , l$. Hence, 
\begin{eqnarray}
[\theta_0,\theta_i] = \frac{\partial \theta_i}{\partial z}\theta_0 - \frac{\partial \theta_0}{\partial z}\theta_i = B^iB^0.
\end{eqnarray}
Hence the higher order Lie bracket 
\[\big[\theta_0,[\theta_0,\theta_i]\big] =[\theta_0,B^iB^0]=0,\]
that is,
\begin{eqnarray}\label{lie_bracket_theta_theta}
\big(ad_{\theta_0}^k,\theta_i\big)=0 \textnormal{ for } k>1.
\end{eqnarray}
Next, we look at the Lie brackets between the drift vector field with $\theta_0$. In particular, 
\begin{eqnarray}
[f_d,\theta_0] = \frac{\partial \theta_0}{\partial z}f_d - \frac{\partial f_d}{\partial z}\theta_0 = -AB^0.
\end{eqnarray}
With this we have,
\begin{eqnarray}\label{lie_bracket_f_theta0}
\big(ad_{f_d}^k,\theta_0\big) = (-1)^kA^kB^0.
\end{eqnarray}
Similarly, 
\begin{eqnarray*}
[f_d,\theta_i] = \frac{\partial \theta_i}{\partial z}f_d - \frac{\partial f_d}{\partial z}\theta_i = B^iAz - AB^iz = [B^i,A]z
\end{eqnarray*}
and
\begin{eqnarray*}
\big(ad_{f_d}^2,\theta_i\big)= \big[f_d,[f_d,\theta_i]\big]= \big[f_d,[B^i,A]z\big] = \big[[B^i,A],A\big]z.
\end{eqnarray*}
With this, one can compute $\big(ad_{f_d}^k,\theta_i\big)$ as
\begin{eqnarray}\label{lie_bracket_f_thetai}
\big(ad_{f_d}^k,\theta_i\big) = \big[[B^i,A],A^{k-1}\big]z.
\end{eqnarray}
Now, since $A\in\mathbb{R}^{n\times n}$, from Cayley-Hamilton theorem one needs to compute $\big(ad_{f_d}^k,\theta_i\big)$ only upto for $k=n$ and then the accessibility matrix $\cal Q$ is 
\begin{eqnarray}\label{accessibility_bilinear}\nonumber
{\cal Q} &=& \textnormal{Lie}(\{f_d, \theta_i,\theta_0\})\\ \nonumber
&=& \Big[B^0, B^1, \dots , B^l, -AB^0, \dots , (-1)^{n-1}A^{n-1}B^0,\\ \nonumber
&& \quad B^1B^0,\dots , B^lB^0, [B^1,A]z,[B^2,A]z,\dots ,\\
&& \qquad \big[[B^l,A]A^{n-1}\big]z \Big]
\end{eqnarray}
Hence, the bilinear system (\ref{bilinear_general}) is locally controllable at $z_0\in\mathbb{R}^d$ if the rank of $\cal Q$ at $z_0$ is $d$ and from theorem \ref{theorem_local_controllability} and definition \ref{def_global_controllability} it is globally controllable if rank${\cal Q}=d$ for all $z\in\mathbb{R}^d$.
\end{proof}

\subsubsection{Controllability of the control-affine system and its Koopman representation}

Given a control-affine system of the form (\ref{affine_sys}) evolving in $\mathbb{R}^d$, lemma \ref{lemma_koopman_representation} gives the equivalent bilinear Koopman representation in $\mathbb{R}^n$. Again, theorem \ref{theorem_bilinear_controllability} gives sufficient conditions for global controllability of a bilinear system. With this, we have the following theorem relating controllability of the control-affine system (\ref{affine_sys}) and its Koopman bilinear representation. 

\begin{theorem}\label{theorem_nonlinear_koopman_controllability}
Consider a control-affine system of the form 
\[x_{t+1} = T(x_t) + g(x_t)u_t,\]
$x_t\in\mathbb{R}^d$ and its bilinear Koopman representation of the form
\[z_{t+1} = \mathbb{U}z_t + u_t B z_t,\]
with $z_t\in\mathbb{R}^d$, such that $\Phi:\mathbb{R}^d\rightarrow\mathbb{R}^n$ are the set of observables. Then if the Koopman bilinear form is controllable, then the nonlinear control-affine system is also controllable.
\end{theorem}
\begin{proof}
We prove this by contrapositive arguments. Suppose that for the nonlinear control-affine system $x_{t+1} = T(x_t) + g(x_t)u_t$, there exists $x_0$ and $x_T$ such that there does not exist any control $u_t$ which can drive the system trajectory from $x_0$ to $x_T$ in time $T$. Let ${\cal T}(x_0,u)$ be the set of all control trajectories that start from $x_0$. Hence 
\[{\cal T}(x_0,u)\cap x_T = \{\varphi\},\]
where $\{\varphi\}$ denotes the empty set. Let $\Phi \big({\cal T}(x_0,u)\big)$ and $\Phi(x_T)$ be the images of ${\cal T}(x_0)$ and $x_T$ under the observables $\Phi$. Since the mapping $\Phi$ is injective, in the lifted space we again have 
\[\Phi \big({\cal T}(x_0,u)\big)\cap \Phi(x_T) = \{\varphi\}.\]
Hence, if the control-affine system is uncontrollable, then the lifted Koopman bilinear system is also uncontrollable. In other words, if the lifted Koopman bilinear system is controllable then so is the original control-affine nonlinear system.
\end{proof}

\subsection{Control Lyapunov Function}

Now that we have related the controllability properties of the control-affine system with those of the lifted Koopman bilinear system, we now address the problem of stabilization of the nonlinear control system. 

% Consider the control affine discrete-time system
% \begin{equation}
% x_{t+1} = T(x_t) + g(x_t)u_t
% \end{equation}
% where $x_t\in\mathbb{R}^d$, $u_t\in\mathbb{R}$ is the single input to the system and $T,g:\mathbb{R}^d\to\mathbb{R}^d$ are at least of class ${\cal C}^2$. We further assume that $x=0$ is as unstable equilibrium point of the uncontrolled system $x_{t+1}=T(x_t)$. 

The control objective is to design a state feedback control law $u_t=k(x_t)$, with $k:\mathbb{R}^d\to\mathbb{R}$, such that the origin $x=0$ is asymptotically stable within some domain $\Omega\in\mathbb{R}^d$ for the closed loop system
\begin{eqnarray}\label{closed_loop_sys}
x_{t+1} = T(x_t) + g(x_t)k(x_t).
\end{eqnarray}

\begin{definition}[Control Lyapunov Function]
Let $V(x_t)$ be a radially unbounded, positive definite function with $V(x_t)>0$, $\forall x_t\neq 0$ and $V(0)=0$. If for any $x_t\in\mathbb{R}^d$, there exist real values $u_t$
such that
\[\Delta V(x_t,u_t)<0,\]
where $\Delta V(x_t,u_t) = V(x_{t+1})-V(x_t)$, then $V(\cdot)$ is called the \emph{discrete-time control Lyapunov function} for the system (\ref{affine_sys}). 
\end{definition}

In the subsequent subsection, we will use the concept of control Lyapunov function to prove the stability of the origin of the closed-loop system.

\subsection{Stabilization of Control-affine Systems}

% We consider $N$ eigenfunctions  of the Koopman operator as 
% \[{\bf \Psi}(x)=[\psi_1(x),\cdots , \psi_N(x)],\]
% with eigenvalues $\lambda_i\in\mathbb{C}$, for $i=1,2,\cdots , N$.
% In this new co-ordinate given by $z={\bf \Psi}(x)$, the system (\ref{affine_sys}) takes the form \cite{surana_cdc_2016}
% \begin{eqnarray}\label{Koopman_representation}
% z_{t+1}={\bf K}z_t + \frac{\partial {\bf \Psi}}{\partial x}g(x_t)u_t,
% \end{eqnarray}
% where ${\bf K}$ is the finite-dimensional approximation of the Koopman operator for the autonomous system $x_{t+1}=T(x_t)$.

% Under assumption (\ref{assumption_1}), the system (\ref{Koopman_representation}) can be written as a bilinear control system
% \begin{eqnarray}\label{bilinear}
% z_{t+1}={\bf K}z_t + u_tBz_t.
% \end{eqnarray}

Given a nonlinear control-affine system (\ref{affine_sys}), we seek a stabilizing control law of the form $u=k(z)$ which stabilizes the origin of the system (\ref{affine_sys}) and we do so by considering the bilinear Koopman representation of the control-affine system. In particular, we seek a state feedback control law $u_t = k(z_t)$ in the lifted space $\mathbb{R}^n$ which quadratically stabilizes the system (\ref{Koopman_representation}) inside the ellipsoid
\[\mathscr{E}=\{z\in\mathbb{R}^n | z^\top Q^{-1}z\leq 1\}, \quad Q \succ 0.\]

The design of stabilizing control law for general bilinear systems has been addressed in the literature and interested readers are referred to \cite{artstein_stabilization,astofli_feedback,xiushan_stabilization,khlebnikov2016quadratic}. The theorem here is stated and proved to suit the framework of the problem at hand.

%\begin{theorem}
%\textcolor{red}{This is the main theorem where we need to prove stabilizability of (\ref{Koopman_representation}). I guess a quadratic Lyapunov function will suffice and we will get some Matrix Inequalities that can be used to formulate an optimization problem to obtain the control Lyapunov function.}
%\end{theorem}

To prove the stabilizability theorem, we need the following lemma \cite{petersen_stabilization}:
\begin{lemma}\label{peterson_lemma}
Let $G=G^\top\in\mathbb{R}^{n\times n}$, $M\in\mathbb{R}^{n\times q}$, $N\in\mathbb{R}^{n}$ and $0\prec P=P^\top\in\mathbb{R}^{q\times q}$. The matrix inequality
\[G + M\delta N^\top + N\delta^\top M^\top\prec 0\]
holds for all $\delta\in\mathbb{R}^q$,  $\delta^\top P \delta \leq 1$ if and only if there exists a real number $\varepsilon$ such that
\[
\begin{pmatrix}
G & M & N\\
M^\top & -\varepsilon P & 0\\
N^\top & 0 & -\frac{1}{\varepsilon}I
\end{pmatrix}\prec 0.\]
where 0 and $I$ denotes matrix with all zeros and identity matrix of appropriate dimensions respectively.
\end{lemma}

With this, we have the following theorem.

\begin{theorem}\label{theorem_stabilizing_control}
Let a positive symmetric matrix $Q=Q^\top>0$ and a vector $y$ be such that the matrix inequality
\begin{eqnarray}\label{deltaV_eqv_theorem}
\begin{pmatrix}
-Q & 0 & y & Q\mathbb{U}^\top\\
0 & -\varepsilon Q & 0 & QB^\top\\
y^\top & 0 & -\frac{1}{\varepsilon}I & 0\\
\mathbb{U} Q & BQ & 0 & -Q
\end{pmatrix}\prec 0
\end{eqnarray}
is satisfied for some $\varepsilon\in\mathbb{R}$. Then the linear feedback with controller gain $k=Q^{-1}y$ stabilizes the bilinear system 
\[z_{t+1}=\mathbb{U}z_t + u_tBz_t\]
inside the ellipsoid 
\[\mathscr{E} = \{z\in\mathbb{R}^n: z^\top Q^{-1}z \leq 1\},\]
and the quadratic form $V(z)=z^\top Q^{-1}z$ is control Lyapunov function for the bilinear system (\ref{Koopman_representation}).

\end{theorem}

\begin{proof}
Consider the quadratic Lyapunov function $V(z)=z^\top P z$ for $P=P^\top >0$. Then
{\small
\begin{eqnarray}\label{deltaV}\nonumber
\Delta V &=& V(z_{t+1})-V(z_t) \\ \nonumber
&=& (\mathbb{U} z_t + u_tBz_t)^\top P (\mathbb{U} z_t + u_tBz_t)  - z_t^\top P z_t\\ \nonumber
&=& z_t^\top (\mathbb{U}^\top P \mathbb{U} - P)z_t +  u z_t^\top (\mathbb{U}^\top P B + B^\top P \mathbb{U})z_t \\
&& \quad + u^2 z_t^\top B^\top PB z_t
\end{eqnarray}
}

For the system (\ref{Koopman_representation}) to be stable, $\Delta V<0$ for $z\neq 0$. We seek a state feedback control law $u_t = k(z_t)$ that stabilizes the origin. Hence, for stabilizability, rewriting (\ref{deltaV}) in matrix form, we have
{\small
\begin{eqnarray}
\begin{pmatrix}
\mathbb{U}^\top P \mathbb{U} - P & k(z)B^\top P\\ 
+ k(z) \mathbb{U}^\top P B + k(z) B^\top P \mathbb{U} \\
PB k(z) & -P
\end{pmatrix}\prec 0.
\end{eqnarray}
}

Under the assumption that we use linear feedback control $u=k^\top z$, we have
{\small
\begin{eqnarray}\label{deltaV_expanded}\nonumber	
\begin{pmatrix}
\mathbb{U}^\top P \mathbb{U} - P & 0\\
0 & -P
\end{pmatrix} + \begin{pmatrix}
\mathbb{U}^\top P B\\
PB
\end{pmatrix}z\begin{pmatrix}
k^\top & 0
\end{pmatrix}+\\
\begin{pmatrix}
k\\
0
\end{pmatrix}z^\top \begin{pmatrix}
B^\top P \mathbb{U} & B^\top P
\end{pmatrix}\prec 0.
\end{eqnarray}
}

The goal is to make the above matrix inequality (\ref{deltaV_expanded}) hold for all $z$ from the ellipsoid
\[\mathscr{E}=\{z\in\mathbb{R}^n: z^\top P z \leq 1\}.\]

Using lemma \ref{peterson_lemma}, equation (\ref{deltaV_expanded}) can be equivalently written as 
\begin{eqnarray}\label{deltaV_eqv}
\begin{pmatrix}
\mathbb{U}^\top P \mathbb{U} - P & 0 & \mathbb{U}^\top P B & k\\
0 & -P & PB & 0\\
B^\top P \mathbb{U} & B^\top P & -\varepsilon P & 0\\
k^\top & 0 & 0 & -\frac{1}{\varepsilon}I
\end{pmatrix}\prec 0.
\end{eqnarray}

The feasibility of (\ref{deltaV_eqv}) implies stability and thus the feedback control law $u=k^\top z$ stabilizes the bilinear control system 
\[z_{t+1}=\mathbb{U}z_t + u_tBz_t.\]

Using Schur's Lemma, from (\ref{deltaV_eqv}) we have

\begin{eqnarray}\label{deltaV_eqv2}
\begin{pmatrix}
\mathbb{U}^\top P \mathbb{U} - P & \mathbb{U}^\top P B & k\\
B^\top P \mathbb{U} & B^\top P -\varepsilon P & 0\\
k^\top & 0 & -\frac{1}{\varepsilon}I
\end{pmatrix}\prec 0
\end{eqnarray}
which can be further written as
{\small
\begin{eqnarray}\label{deltaV_eqv3}
\begin{pmatrix}
-P & 0 & k\\
0 & -\varepsilon P & 0\\
k^\top & 0 & -\frac{1}{\varepsilon}I
\end{pmatrix} + \begin{pmatrix}
\mathbb{U}^\top\\
B^\top\\
0
\end{pmatrix}P \begin{pmatrix}
\mathbb{U} & B & 0
\end{pmatrix}\prec 0.
\end{eqnarray}
}

Applying Schur's Lemma to (\ref{deltaV_eqv3}), we have
\begin{eqnarray}\label{deltaV_eqv4}
\begin{pmatrix}
-P & 0 & k & \mathbb{U}^\top\\
0 & -\varepsilon P & 0 & B^\top\\
k^\top & 0 & -\frac{1}{\varepsilon}I & 0\\
\mathbb{U} & B & 0 & -P^{-1}
\end{pmatrix}\prec 0.
\end{eqnarray}

Let $Q = P^{-1}$ and pre- and post-multiplying (\ref{deltaV_eqv4}) by $\textnormal{diag}(Q,Q,I,I)$, we have
\begin{eqnarray}\label{deltaV_eqv5}
\begin{pmatrix}
-Q & 0 & Qk & Q\mathbb{U}^\top\\
0 & -\varepsilon Q & 0 & QB^\top\\
k^\top Q & 0 & -\frac{1}{\varepsilon}I & 0\\
\mathbb{U} Q & BQ & 0 & -Q
\end{pmatrix}\prec 0.
\end{eqnarray}
Setting $y=Qk\in\mathbb{R}^n$, we have the proof.
\end{proof}

The above theorem proves the existence of a state feedback stabilizing control law that stabilizes the origin of a bilinear control system. However, we can do more. In particular, we can strive to maximize the stabilizability ellipsoid. One way to do it is to maximize the volume of the stabilizability ellipsoid. This can be achieved in the following way:

\begin{corollary}\label{main_corollary}
Consider the optimization problem:
\begin{eqnarray}\label{ellipsoid_max}
\begin{aligned}
&\qquad\quad \max \log \det Q\\
\textnormal{subject to } & \begin{pmatrix}
-Q & 0 & y & Q\mathbb{U}^\top\\
0 & -\varepsilon Q & 0 & QB^\top\\
y^\top & 0 & -\frac{1}{\varepsilon}I & 0\\
\mathbb{U} Q & BQ & 0 & -Q
\end{pmatrix}\prec 0    
\end{aligned}
\end{eqnarray}
with respect to the optimization variables $Q=Q^\top\in\mathbb{R}^{n\times n}$ and $y\in\mathbb{R}^n$. Then the ellipsoid 
\[\hat{\mathscr{E}} = \{  z\in\mathbb{R}^n: z^\top Q^{-1}z \leq 1\}\] is the stabilizability ellipsoid of the system (\ref{Koopman_representation}) with the feedback control law given by $u=\hat{k}^\top z$, where $\hat{k}=\hat{Q}^{-1}\hat{y}$.
\end{corollary}

Though the optimization problem (\ref{ellipsoid_max}) of corollary (\ref{main_corollary}) maximizes the stabilizability ellipsoid, the objective function is optimized under a strict constraint. To make the problem well-posed, we make the strict inequality a non-strict inequality as follows:
\begin{eqnarray}
\begin{pmatrix}
-\theta Q & 0 & y & Q\mathbb{U}^\top\\
0 & -\varepsilon Q & 0 & QB^\top\\
y^\top & 0 & -\frac{1}{\varepsilon}I & 0\\
\mathbb{U} Q & BQ & 0 & -Q
\end{pmatrix}\preceq 0   
\end{eqnarray}
where $0 < \theta < 1$. Hence, the feedback control law which maximizes the stabilizability ellipsoid is obtained as the solution of the following optimization problem:
\begin{eqnarray}\label{main_opt}
\begin{aligned}
&\qquad\quad \max \log \det Q\\
\textnormal{subject to } & \begin{pmatrix}
-\theta Q & 0 & y & Q\mathbb{U}^\top\\
0 & -\varepsilon Q & 0 & QB^\top\\
y^\top & 0 & -\frac{1}{\varepsilon}I & 0\\
\mathbb{U} Q & BQ & 0 & -Q
\end{pmatrix}\preceq 0.
\end{aligned}
\end{eqnarray}

Theorem \ref{theorem_stabilizing_control} provides a state feedback control law that stabilizes the origin of a bilinear control system , which in our case is the lifted Koopman system. Now, the mapping $\Phi$ is such that $\Phi(0)=0$. Hence the stabilized trajectories of the lifted bilinear Koopman system, which converge to the origin of the lifted space $\mathbb{R}^n$, when projected back to the original state space $\mathbb{R}^d$, they too converge to the origin of the state space. Thus the state feedback control law stabilizes the origin of the nonlinear control-affine system.

\section{Extended Dynamic Mode Decomposition (EDMD)}\label{section_DMD}

Typically, the Koopman operator $(\mathbb{U})$ for a dynamical system is infinite-dimensional and hence for data-driven computations, a finite-dimensional approximation of the Koopman operator is computed from the time-series data. Moreover, satisfying assumptions \ref{assumption_span} and \ref{assumption_1} is difficult. In this section, we briefly describe the procedure for computing the approximate system matrices of the lifted bilinear system from time-series data.

Consider snapshots of time-series data  
\[X = [x_1, x_2, \cdots , x_{M}], \quad Y = [y_1, y_2, \cdots , y_{M}],\]
obtained from simulating a discrete time dynamical system $y_i = T(x_i)$ or from an experiment, where $x_i$ and $y_i$ are consecutive data points. Let $x_i,y_i\in\mathbb{R}^d$. Let $\Phi=
\{\phi_1,\phi_2,\ldots,\phi_n\}$ be the set of dictionary functions or observables, where $\phi_i : \mathbb{R}^d \to \mathbb{C}$. Let ${\cal G}_{\Phi}$ denote the span of ${\Phi}$ such that ${\cal G}_{\Phi}\subset L_2(X,{\cal B},\mu)$. One important observation is the fact that for a good approximation of the Koopman operator, the set of dictionary functions should be able to approximate the leading eigenfunctions of the Koopman operator. 

With this, if $\bf K$ is the finite dimensional approximation of the Koopman operator, then the  matrix $\bf K$ can be obtained as a solution to the following least square problem \cite{EDMD_williams}. 
\begin{equation}\label{edmd_op}
\min\limits_{\bf K}\parallel {\bf G}{\bf K}-{\bf A}\parallel_F
\end{equation}
\begin{eqnarray}\label{edmd1}
\begin{aligned}
& {\bf G}=\frac{1}{M}\sum_{m=1}^M {\Phi}({x}_m)^\top {\Phi}({x}_m)\\
& {\bf A}=\frac{1}{M}\sum_{m=1}^M {\Phi}({x}_m)^\top {\Phi}({y}_{m}),
\end{aligned}
\end{eqnarray}
with ${\bf K},{\bf G},{\bf A}\in\mathbb{C}^{n\times n}$. The optimization problem (\ref{edmd_op}) can be solved explicitly to obtain following solution for the matrix $\bf K$, i.e.,
\begin{eqnarray}\label{K_EDMD}
{\bf K}_{EDMD}={\bf G}^\dagger {\bf A},\label{EDMD_formula}
\end{eqnarray}
where ${\bf G}^{\dagger}$ is the Moore-Penrose inverse of matrix $\bf G$.
Dynamic Mode Decomposition (DMD) is a special case of EDMD algorithm with ${\Phi}(x) = x$.

Now, within the setting of this paper, since the Koopman representation of the nonlinear system (\ref{affine_sys}) is a bilinear system of the form (\ref{Koopman_representation}), the finite-dimensional representation of the control-affine system (\ref{affine_sys}) will take the form
\begin{eqnarray}\label{finite_bilinear}
z_{t+1} = Az_t + u_tBz_t.
\end{eqnarray}

It is assumed that the leading Koopman eigenfunctions are contained inside ${\cal G}_{\Phi}$, and the eigenvalues of ${\bf K}_{EDMD}$ are approximations of the Koopman eigenvalues. Let the right eigenvectors of Koopman ${\bf K}_{EDMD}$ be denoted by $v_j$, $j = 1,\hdots,n$. Then the Koopman eigenfunctions are approximated using the right eigenvectors and the dictionary as 
\begin{align}
    h_j & = {\Phi}^{\top} v_j, \quad j = 1,2,\dots,n, \label{eq:Koopman_EF}
\end{align}
where $h_j$ is the approximation of the eigenfunction of Koopman operator corresponding to the $j^{th}$ eigenvalue. For the bilinear system (\ref{finite_bilinear}), the system matrix $A$ on the lifted space is expressed as a block diagonal matrix of Koopman eigenvalues $\lambda_1, \lambda_2, \dots, \lambda_n$ where 
\begin{enumerate}
    \item ${A}_{(i,i)} = \lambda_i$ if $h_i$ is real, and 
    \item if $h_i$ and $h_{i+1}$ are complex conjugate pairs, then $\begin{bmatrix} {A}_{i,i} & {A}_{i,i+1} \\ {A}_{i+1,i} & {A}_{i+1,i+1} \end{bmatrix}$ = $\vert \lambda_i \vert \begin{bmatrix} \cos(\angle \lambda_i) & \sin(\angle \lambda_i) \\ -\sin(\angle \lambda_i) & \cos(\angle \lambda_i) \end{bmatrix}$
\end{enumerate}
From \eqref{eq:Koopman_EF}, while the Koopman eigenfunctions are real-valued, we have, 
\begin{align*}
    H(x) = V^{\top} {\bf \Psi} (x), 
\end{align*}
where 
\begin{align*}
&V = \begin{bmatrix} v_1 & v_2 & \hdots & v_n \end{bmatrix} \\
&H(x) = \begin{bmatrix} h_1(x) & h_2(x) & \dots & h_n(x) \end{bmatrix}.
\end{align*}
We now approximate the matrix $B \in \mathbb{R}^{n\times n}$ as 
\begin{align*}
    \frac{\partial H}{\partial x} g(x) & = V^{\top} \frac{\partial {\Phi}}{\partial x} g(x) \\
    & = B H(x) = B V^{\top} {\Phi} (x) = \tilde B {\Phi} (x).
\end{align*}
Under the assumption that the observable functions ${\Phi}$ is a collection of monomial functions, the terms $\frac{\partial {\Phi}}{\partial x}$ lies in the span of ${\Phi}(x)$. Furthermore, the eigenvector matrix, $V$ is invertible and hence the $\tilde B$ can be found explicitly. Finally, the matrix $B$ can be computed such that $B = \tilde B (V^{\top})^{-1}$. 

% \section{Design of Koopman Based Predictor}\label{section_predictor}
% % The Koopman operator generates a linear system in a higher-dimensional space, even if the underlying system is nonlinear. 
% The linearity of the operator enables the design of linear predictors for nonlinear systems. The following is presented briefly for the self-containment of the paper; for details the readers are referred to \cite{korda_mezic_predictor}. Let $\{x_0,\ldots,x_M\}$ be the training data-set and $\bf K$  be the finite-dimensional approximation of the transfer Koopman operator obtained using (\ref{K_EDMD}). Let $\bar x_0$ be the initial condition from which the future is to be predicted. The initial condition from state space is mapped to the feature space using the same choice of basis function used in the robust approximation of Koopman operator i.e., \[\bar x_0\implies {\bf \Psi}(\bar x_0)^\top=: {\bf z}\in \mathbb{R}^K.\] This initial condition is propagated using Koopman operator as, \[{\bf z}_n={\bf K}^n{\bf z}.\]
% The predicted trajectory in the state space is then obtained as, 
% \[\bar x_n=C {\bf z}_n,\]
% where matrix $C$ is obtained as the solution of the following least squares problem:
% \begin{eqnarray}\label{C_pred}
% \min_C\sum_{i = 1}^M \parallel x_i - C \boldsymbol \Psi (x_i)\parallel_2^2.
% \end{eqnarray}

\section{Simulation Results}\label{section_simulation}
This section discusses two nonlinear dynamical systems where we stabilize the origin using our proposed approach. In particular, we stabilize the origin for the Van der Pol oscillator and the Henon map. 

\subsection{Van der Pol Oscillator}

The Van der Pol oscillator is a non-conservative dynamical system with nonlinear damping whose equations of motion are given by
\begin{eqnarray}\label{van_der_pol}
\begin{aligned}
& \dot{x} = y\\
& \dot{y} = \mu (1-x^2)y -x.
\end{aligned}
\end{eqnarray}

The constant $\mu\geq 0$ controls the damping and when $\mu=0$, one recovers the simple harmonic oscillator, which is a conservative system. For $\mu > 0$, the system exhibits stable limit cycle oscillations, with the origin being an unstable equilibrium point. The invariant measure for the Van der Pol oscillator is shown in Fig. \ref{van_der_pol_inv_meas_fig}. 

\begin{figure}[htp!]
\centering
\includegraphics[scale=.225]{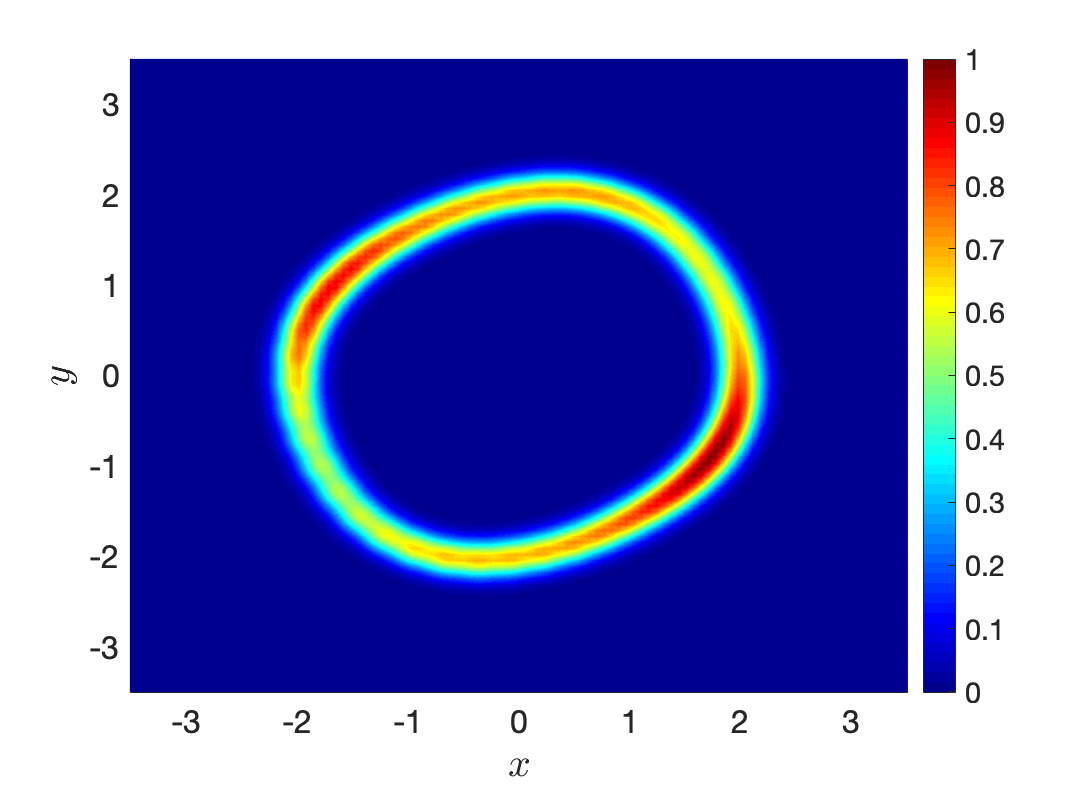}
\caption{Invariant measure of the unforced Van der Pol oscillator.}\label{van_der_pol_inv_meas_fig}
\end{figure}

The goal is to stabilize the origin by designing a state feedback law. To this end, we assume the following controlled Van der Pol oscillator:
\begin{eqnarray}\label{control_van_der_pol}
\begin{aligned}
& \dot{x} = y\\
& \dot{y} = \mu (1-x^2)y -x + u,
\end{aligned}
\end{eqnarray}
where $u$ is the scalar control input. Comparing (\ref{control_van_der_pol}) with (\ref{affine_sys}), we have $g(x,y)=[0\quad 1]^\top$. 

For simulation purposes, the constant $\mu$ was set to one and data was collected for 10 seconds with discretization step $\delta t = .01$ seconds. Furthermore, as dictionary functions, we used monomials up to order 5. Hence cardinality of the dictionary function set $(\Phi)$ is 21, so that the Koopman operator ${\bf K}\in\mathbb{R}^{21\times 21}$. Once the Koopman operator for the open-loop system is computed, the optimization problem (\ref{main_opt}) computes the optimal stabilizability ellipsoid $(\hat{Q})$ and the vector $y$. For simulation purposes, the parameter $\theta$ and $\varepsilon$ were set to 0.001 and 0.01 respectively. With this, the feedback control law gain $(k)$ in the lifted space is computed as $k = \hat{Q}^{-1} y$ and the feedback control law in the lifted space is given by $k^\top z$, where $z={\Phi}(x,y)$.

\begin{figure}[htp!]
\centering
\includegraphics[scale=.45]{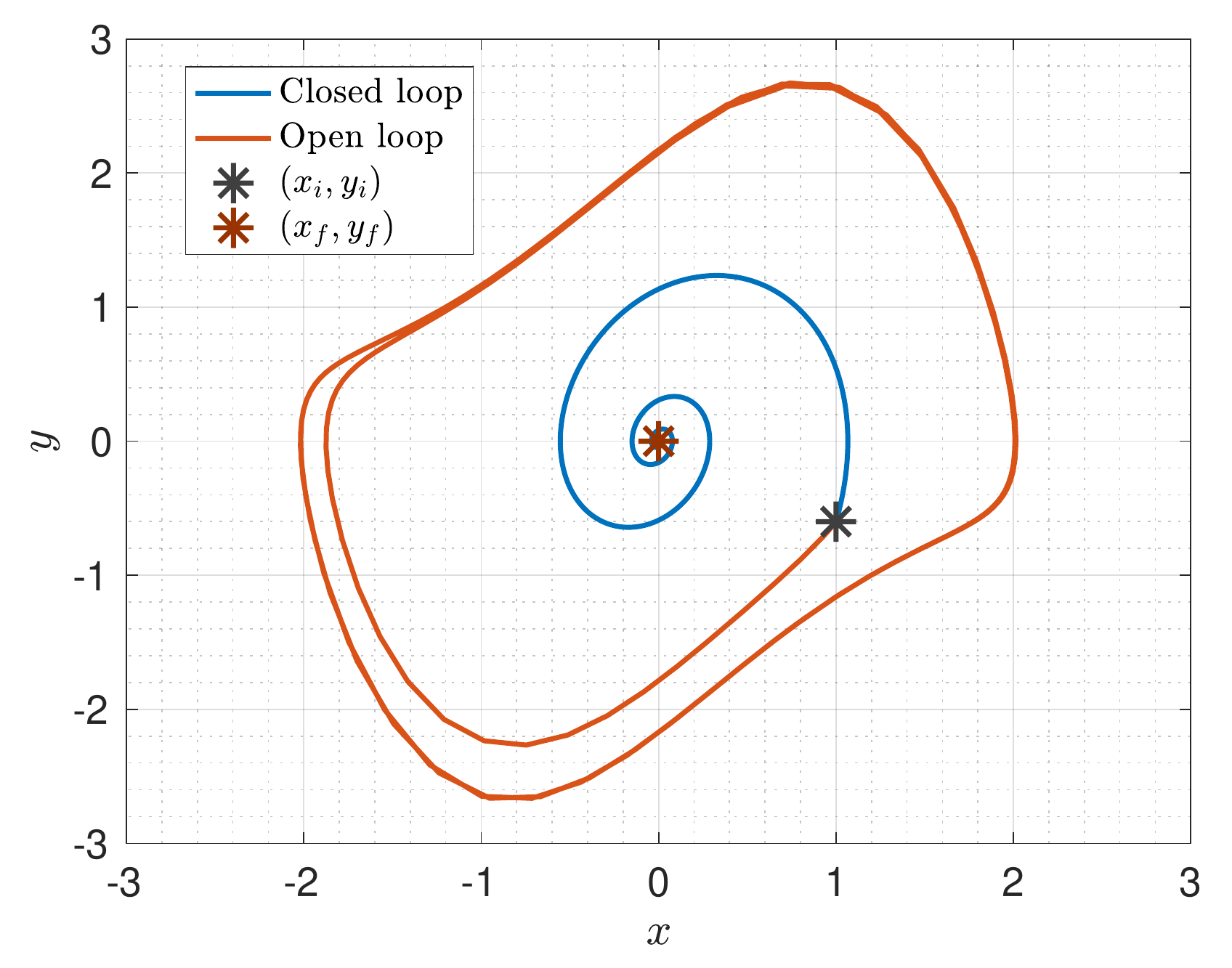}
\caption{Open loop and closed loop trajectory for the Van der Pol oscillator. Without any control the open loop trajectory settles on the stable limit cycle, while with the state feedback control, the closed loop trajectory goes to the origin.}\label{van_der_pol_fig}
\end{figure}

Fig. \ref{van_der_pol_fig} shows the open and closed-loop trajectories for the Van der Pol oscillator starting from a random initial condition. We observe that the trajectory settles to the stable limit cycle attractor without the control input, whereas with the control law, the closed-loop trajectory is stabilized to the origin.

\subsection{Henon Map}

The Henon map is a discrete-time dynamical system which exhibits chaotic behaviour. The equations of motion are
\begin{eqnarray}\label{henon_eq}
\begin{aligned}
& x_{n+1} = 1 - ax_n^2 + y_n\\
& y_{n+1} = bx_n.
\end{aligned}
\end{eqnarray}

For the classical Henon map,  with constants $a = 1.4$ and $b=0.3$, the Henon map has a chaotic attractor, as shown by the invariant measure in Fig. \ref{henon_map_inv_meas_fig}. 
\begin{figure}[htp!]
\centering
\includegraphics[scale=.225]{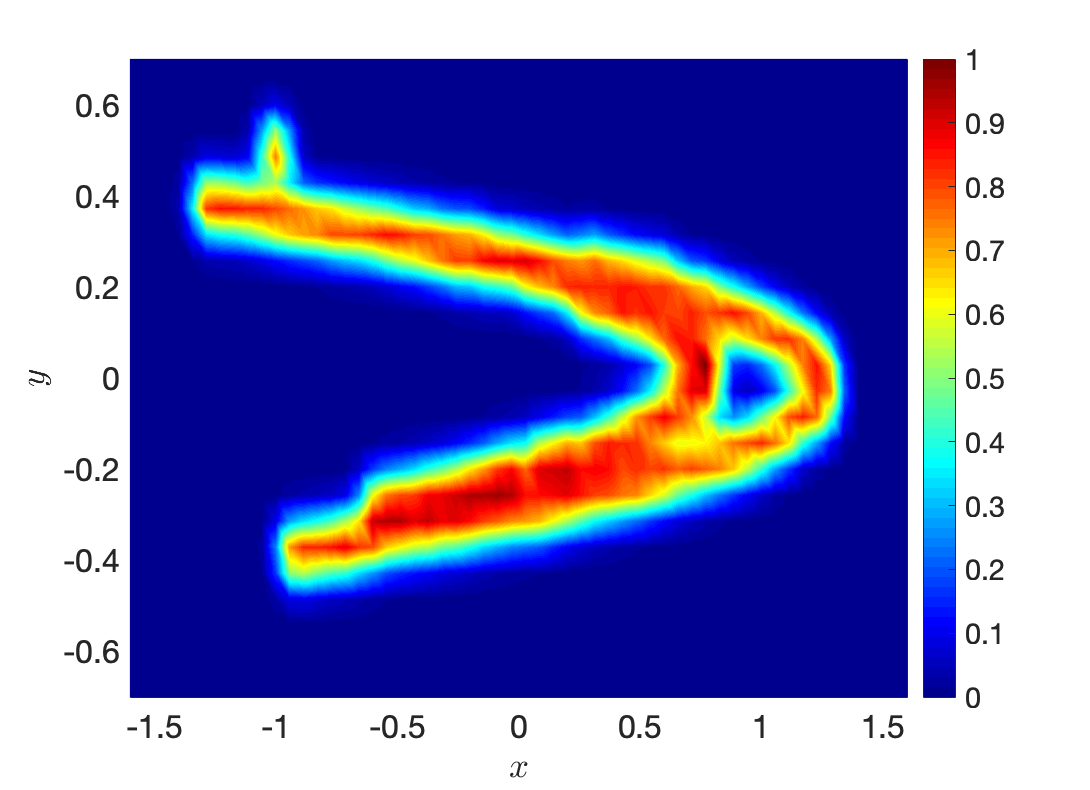}
\caption{Invariant measure of the unforced Henon map.}\label{henon_map_inv_meas_fig}
\end{figure}

We assume that the control map $g(x,y)=[0\quad 1]^\top$ so that the controlled Henon map is given by

\begin{eqnarray}\label{henon_control_eq}
\begin{aligned}
& x_{n+1} = 1 - ax_n^2 + y_n\\
& y_{n+1} = bx_n + u.
\end{aligned}
\end{eqnarray}

\begin{figure}[htp!]
\centering
\includegraphics[scale=.22]{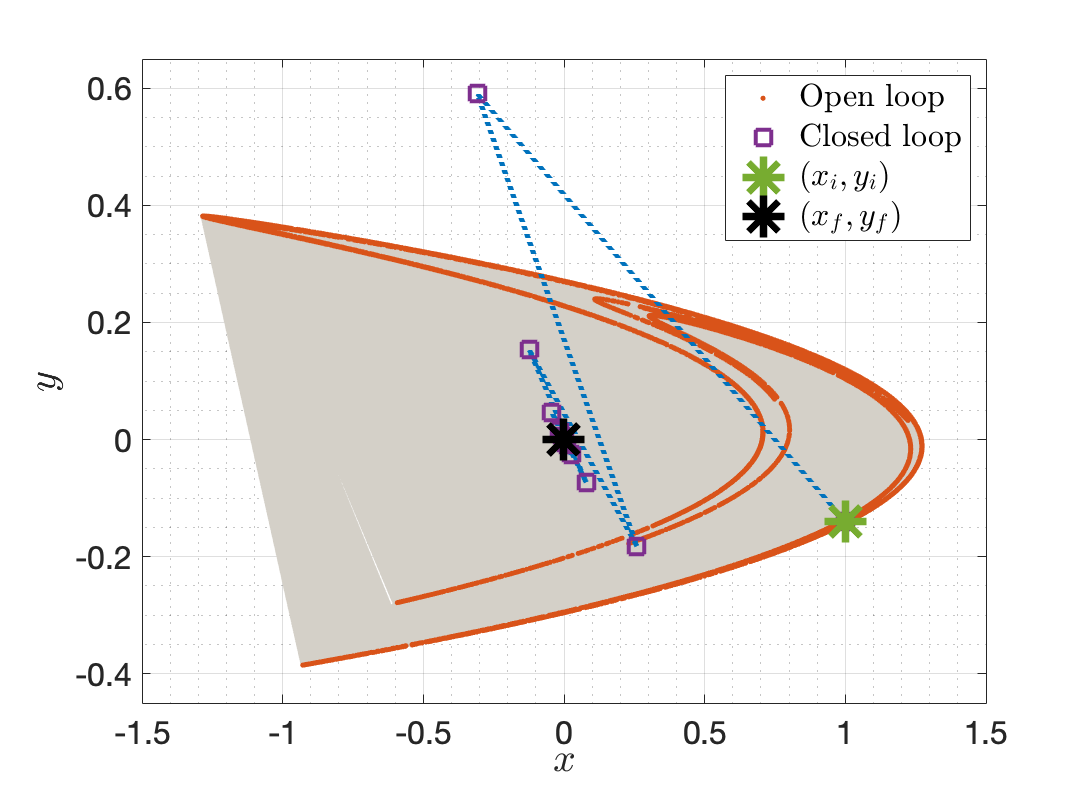}
\caption{Open loop and closed loop trajectory for the Henon map. Without any control the open loop trajectory settles on the chaotic attractor, while with the state feedback control, the closed loop trajectory goes to the origin.}\label{henon_map_fig}
\end{figure}
The data for the uncontrolled Henon map (\ref{henon_eq}) was generated for 10000-time steps and the Koopman operator was computed using monomials up to degree 2. Hence the cardinality of the dictionary function set $\Phi$ is five so that the Koopman operator ${\bf K}\in\mathbb{R}^{5\times 5}$. The optimization problem (\ref{main_opt}) yields the state feedback control law. Subsequently, in Fig. \ref{henon_map_fig}, we show both the open and closed-loop trajectory, starting from a random initial condition. It can be seen that the trajectory for the uncontrolled system settles on the chaotic attractor. In contrast, the trajectory evolves to the origin of the controlled system, thus stabilizing the origin. However, one may note that the closed-loop trajectory is non-smooth, unlike the Van der Pol oscillator, and this is due to the nature of the trajectories of the uncontrolled Henon map.

\section{Conclusions}\label{section_conclusion}

This paper develops a data-driven method for designing a stabilizing control law for general discrete-time control-affine systems. In particular, we use the Koopman operator framework to lift the dynamical system to a higher dimensional space where the evolution is given by a bilinear system. Before the design of the state feedback control law, we analyzed the controllability of the lifted bilinear system and related it to the controllability of the original nonlinear control-affine system. With this, we designed a state feedback stabilizing control law in the higher dimensional space for the bilinear system and proved that this state-feedback law quadratically stabilizes the origin of the bilinear control system in the higher dimensional space. Furthermore, we proposed an optimization problem that maximizes the stabilizability ellipsoid. The proposed approach is demonstrated to stabilize the origin of the Van der Pol oscillator (a nonlinear system with a stable limit cycle) and the Henon map (a chaotic system) from time-series data. 

% use our approach to stabilize the origin of the Van der Pol oscillator (a nonlinear system with a stable limit cycle) and the Henon map (a chaotic system).

\bibliographystyle{IEEEtran}
\bibliography{Refs,reference,subhrajit_data_driven_control,subhrajit_gramian}

\end{document}